\documentclass[12pt, titlepage]{amsart}
\usepackage{lineno}
\usepackage{amsmath}
\usepackage{amsthm}
\usepackage{xr}
\usepackage{pgfplots}
\graphicspath{{./art/}}
\usetikzlibrary{plotmarks}
\newtheorem{proposition}{Proposition}
\newtheorem{theorem}{Theorem}
\newtheorem{definition}{Definition}
\newtheorem{remark}{Remark}
\title{Supplementary material for \textit{Directional replicability: when can the factor of 2 be omitted}}
\date{}

\usepackage{amsaddr}

\usepackage[normalem]{ulem}
\usepackage[utf8]{inputenc}
\usepackage[T1]{fontenc}
\usepackage{newtxtext}
\usepackage[subscriptcorrection]{newtxmath}

\usepackage{natbib}
\usepackage[plain,noend]{algorithm2e}

\makeatletter
\renewcommand{\algocf@captiontext}[2]{#1\algocf@typo. \AlCapFnt{}#2} 
\def\@algocf@capt@plain{top}
\renewcommand{\algocf@makecaption}[2]{%
  \addtolength{\hsize}{\algomargin}%
  \sbox\@tempboxa{\algocf@captiontext{#1}{#2}}%
  \ifdim\wd\@tempboxa >\hsize
    \hskip .5\algomargin%
    \parbox[t]{\hsize}{\algocf@captiontext{#1}{#2}}
  \else%
    \global\@minipagefalse%
    \hbox to\hsize{\box\@tempboxa}
  \fi%
  \addtolength{\hsize}{-\algomargin}%
}
\makeatother

\begin{document}

\title{Directional replicability: when can the factor of two be omitted}

\author{Vera Djordjilovi\'c}
\address[1]{Department of Economics, Ca' Foscari University of Venice, Italy}
\email{vera.djordjilovic@unive.it}
\author{Tamar Sofer}
\address[2]{Cardiovascular Institute, Beth Israel Deaconess Medical Center, Harvard Medical School, Boston, MA, USA\\
Department of Biostatistics, Harvard T.H. Chan School of Public Health, Boston, MA, USA\\
Division of Sleep Medicine, Brigham and Women’s Hospital, Boston, MA, USA}
\author{Jonathan M. Dreyfuss}
\address[3]{Bioinformatics \& Biostatistics Core, Joslin Diabetes Center,\\ Harvard Medical School, Boston, MA, USA }
\maketitle
\begin{abstract}
Directional replicability addresses the question of whether an effect studied across
$n$ independent studies is present with the same direction in at least
$r$ of them, for $r \geq 2$. When the expected direction of the effect is not specified in advance, the state of the art recommends  assessing replicability separately by combining one-sided $p$-values for both directions (left and right), and then doubling the smaller of the two resulting combined $p$-values to account for multiple testing. In this work, we show that this multiplicative correction is not always necessary, and give 
conditions under which it can be safely omitted. 

\end{abstract}

\section{Introduction}

Low replicability of scientific findings observed in medicine \citep{ioannidis2005most}, economics   \citep{camerer2016evaluating} and psychology \citep{open2015estimating}, has motivated  great interest in developing formal statistical methods for evaluating replicability. We refer the interested reader to a recent review of statistical methodology for replicability analysis by \cite{bogomolov2023replicability}.

Consider a certain phenomenon that is studied in $n$ independent studies. We assume that the object of interest can be represented with a scalar that we refer to as the \textit{effect size} and we  let $\theta = (\theta_1,\ldots,\theta_n) \in \mathbb{R}^n$ denote a vector of true effect sizes across studies. Let further $n^+ = |\left\{i: \theta_i >0\right\}|$ represent the number of positive elements of $\theta$, and similarly $n^-= |\left\{i: \theta_i <0\right\}|$, the number of negative elements of $\theta$. In this work, we are interested in  testing the $r$ out of $n$  directional replicability null hypothesis, denoted by $H_{r/n}$, defined as
\begin{equation*}
 H_{r/n}: n^+ < r \land  n^-<r,
\end{equation*}
for a given $r\leq n$. We consider a general alternative $K_{r/n}: n^+ \geq r \lor n^-\geq r$. Rejecting $H_{r/n}$ allows one to conclude that there are at least $r$ effects of the same sign and thus claim $r$ out of $n$ directional replicability.  Directional replicability imposes a stronger requirement than merely observing an effect in at least $r$ studies (i.e. $n^++n^-\geq r$) since the latter does not require consistency in the sign of the effect. However, rejecting hypothesis $H_{r/n}$ does not imply complete consistency: consider the case $n=4, r=3$ with $n^+=3$, and $ n^-=1$. Then, the effect is  positive  in three studies and $H_{r/n}$ is false; nevertheless, the  effect is negative in the forth study so  the sign of the effect is not consistent across studies. 

The standard approach to testing $H_{r/n}$ is outlined in \cite{owen2009karl}. The hypothesis $H_{r/n}$ is seen as the intersection of the two unilateral  replicability  hypotheses $H_{r/n}^+: n^+ < r$ and $H_{r/n}^{-}: n^- < r$. Each of the unilateral hypotheses is  tested at the significance level $\alpha/2$, for a given $\alpha \in (0,1)$, and the intersection hypothesis $H_{r/n}$ is rejected if either of the two  is rejected. 

The two unilateral replicability null hypotheses $H_{r/n}^+$ and $H_{r/n}^{-}$ can be tested with any test statistic suitable for testing partial conjunction hypothesis. \cite{benjamini2008screening} provide a general procedure for constructing valid test statistics starting from $p$-values of individual studies. See also \cite{bogomolov2023replicability} and \cite{wang2022detecting} for a discussion of $p$-values for partial conjunction hypotheses. In what follows, we provide a brief overview essential for presenting our main result. 

Let us consider a collection of $n$ component null hypotheses $\left\{H_1^+, \ldots, H_n^+\right\}$, where $H_i^+: \theta_i \leq 0$ and the associated alternative is $K_i^+: \theta_i>0$. In this work, we assume to have  independent normal estimators of components of $\theta$, i.e. we let $T_i \sim N(\theta_i,1)$ denote an estimator of $\theta_i$, where, for simplicity, it is assumed that its variance is known and equal to 1.  Then $p_i = 1-\Phi(T_i)$ is a valid $p$-value for $H_i^+$, with $\Phi$ denoting a cumulative distribution function of the standard normal distribution. Independence of $T_i, i=1,\ldots,n$ follows from the assumption of the independence of studies. 

\cite{benjamini2008screening} show that, since $H_{r/n}^+$ is true  if there are at most $r-1$ arbitrarily large positive effects, a valid $p$-value is obtained by ignoring the $r-1$ smallest $p$-values and combining the remaining $n-r+1$  right sided $p$-values with a combining function that leads to a $p$-value stochastically larger or equal to the uniform distribution under $H^+_{r/n}$. Various combining functions could be employed. A simple method for combining $p$-values under arbitrary dependence is the Bonferroni method. Let $p_{(1)},\ldots, p_{(n)}$ be a sequence of $p$-values in a non-decreasing order. The Bonferroni method leads to the following combined $p$-value. 

\begin{definition} Bonferroni partial conjunction $p$-value is $p_{r/n}^+ = (n-r+1)p_{(r)}$.
\end{definition}
 The correction factor of $(n-r+1)$ in Definition 1 corresponds to the usual Bonferroni correction applied to a collection of $n-r+1$ $p$-values obtained by ignoring the $r-1$ smallest $p$-values.  It may be instructive to consider a special case of  a global null hypothesis $H_{1/n}^+$ in which the above correction reduces to the well known factor $n$, and $H_{1/n}^+$ is rejected if $np_{(1)} \leq \alpha$, or equivalently, if at least one $p$-value is below the threshold $\alpha/n$.

Consider now the collection $\left\{H_1^-,\ldots, H_n^-\right\}$, where $H_i^-: \theta_i \geq 0$ with the associated alternative $K_i^-: \theta_i<0$. Then, given the continuity of the distribution of $T_i$,  the $p$-value for $H_i^-$ is $q_i= \Phi(T_i) = 1-p_i$.  As a consequence, the Bonferroni $p$-value for $H_{r/n}^-$ is $p_{r/n}^-= (n-r+1)q_{(r)}=(n-r+1)p_{(n-r+1)}.$ We  state the main result in the following section.

\section{Directional replicability when $r$ is large}

In general, the  procedure for testing $H_{r/n}$ is based on a double of the smaller $p$-value pertaining to unilateral replicability hypotheses, i.e. in this case $p_{r/n}= 2\min\left\{p_{r/n}^-, p_{r/n}^+\right\}$, see \cite{owen2009karl}, \cite{wang2022detecting} and \cite{jaljuli2023quantifying}. The following
Theorem  indicates that when $r$ is large enough with respect to $n$, and the combining function is Bonferroni, the correction factor of two is unnecessary.

\begin{theorem}
\label{theorem_one}
Consider $r$ such that $(n+1)/2 < r\leq n$. Let $p_{r/n}^-$ and $p_{r/n}^+$ be Bonferroni $p$-values for testing $H_{r/n}^-$ and $H_{r/n}^+$, respectively.
Then  $p_{r/n} = \min\left\{p_{r/n}^-,          p_{r/n}^+\right\}$ is a valid $p$-value for $H_{r/n}$. 
\end{theorem}
\begin{proof}
 Consider $\alpha \in (0,1/2)$.    We will show that a test that rejects $H_{r/n}$ if $\min\left\{p_{r/n}^+,  p_{r/n}^-\right\}\leq \alpha$ is an $\alpha$ level test for $H_{r/n}$. 

Let $T_{(1)}, \ldots, T_{(n)}$ be a sequence of estimators of components of $\theta$ in a non-decreasing order. Then the event $[p_{r/n}^+ \leq \alpha]$ occurs if and only if the event $[T_{(n-r+1)}\geq t]$ occurs, where $t=\Phi^{-1}\{1-\alpha/(n-r+1)\}$.  Similarly, the event $[p_{(r/n)}^-\leq \alpha]$ occurs if and only if the event $[T_{(r)}\leq -t]$ occurs. 

Let $\Theta_0\subset \mathbb{R}^n$ be the null parameter space containing all values of $\theta$ such that the hypothesis $H_{r/n}$ is true. Let 
\begin{equation}\label{type1}
c(\theta) = pr_\theta(  T_{(r)}\leq -t \cup T_{(n-r+1)}\geq t)
\end{equation}
indicate the probability of rejecting $H_{r/n}$. We need to show that $\sup_{\theta \in \Theta_0} c(\theta)\leq \alpha$. 

Since $2r > n+1$, then $r> n-r+1$, so that $T_{(r)}\geq T_{(n-r+1)}$ and the two events in \eqref{type1} leading to Type I error  are disjoint. We thus have $c(\theta) =  pr_\theta(T_{(r)}\leq -t) + pr_\theta(T_{(n-r+1)}\geq t)$.
The probability of Type I error can also be expressed as 
\begin{equation*}\label{typeI}
c(\theta) = pr_\theta(X\geq r) + pr_\theta(Y\geq r),
\end{equation*}
where $X = \sum_{i=1}^nX_i$ and $X_i = [T_i\leq -t]$, $i=1,\ldots,n$ are independent indicator variables, and analogously $Y=\sum_{i=1}^nY_i$, $Y_i = [T_i\geq t]$, $i=1,\ldots,n$. 

Since $pr_\theta(X\geq r)=pr_\theta(X_1=1, \sum_{i=2}^nX_i\geq r-1) + pr_\theta(X_1=0, \sum_{i=2}^nX_i\geq r)$ and $pr_\theta(X_1=1) = \Phi(-t-\theta_1)$, the first partial derivative of $pr_\theta(X\geq r) $ is
$$
\frac{{\rm d}pr_\theta(X\geq r)}{{\rm d}\theta_1} = -f(-t-\theta_1)pr_{\theta_{2:n}}\left(\sum_{i=2}^nX_i=r-1\right), 
$$
where $f$ denotes density of the standard normal distribution and $\theta_{i:n}$  denotes a subvector of $\theta$ obtained after removing the first $i-1$ components. 
Analogously, the  partial derivative of $pr_\theta(Y\geq r)$ is
$$
\frac{{\rm d}pr_\theta(Y\geq r)}{{\rm d}\theta_1} = f(t-\theta_1)pr_{\theta_{2:n}}\left(\sum_{i=2}^nY_i=r-1\right).
$$
Consider $\Theta_0^b = \left\{\theta \in \mathbb{R}^n: \theta _i \geq 0, \, i=1,\ldots,r-1, \theta_i=0,\, i=r,\ldots,n\right\} \subset \Theta_0$, a subset of the null parameter space containing all parameter points in which the first $r-1$ components are non-negative and the remaining are zero. For each $\theta\in \Theta_0^b$, random variable $\sum_{i=2}^n Y_i$  stochastically dominates $\sum_{i=2}^n X_i$ and since probability mass functions of sums of independent Bernoulli variables are unimodal \citep{wang1993number}, their respective probability mass functions cross at most once. The point of crossing must be below the mode of the distribution of $\sum_{i=2}^nY_i$, that we denote by $m$. Applying the result of  \cite{samuels1965number} concerning the most likely number of successes, we obtain  $m\leq r-2$.   Then $r-1>m$  and thus $pr_{\theta_{2:n}}\left(\sum_{i=2}^n Y_i=r-1\right) \geq pr_{\theta_{2:n}}\left(\sum_{i=2}^nX_i=r-1\right)$,
with equality holding if and only if $\theta =0$. Since $|t-\theta_1| < |-t-\theta_1|$, the value of the density of the standard normal distribution  at $t-\theta_1$ is higher than the value at $-t-\theta_1$. Therefore, the first partial derivative of $c(\theta)$ with respect to $\theta_1$ is positive on $(0,\infty)$:
$$
\frac{{\rm d}c(\theta)}{{\rm d}\theta_1}=\frac{{\rm d}pr_\theta(X\geq r)}{{\rm d}\theta_1} + \frac{{\rm d}pr_\theta(Y\geq r)}{{\rm d}\theta_1}>0.
$$
The same holds for  partial derivatives with respect to $\theta_i$, $i=2,\ldots,r-1$. As a consequence, $c(\theta)$ is a differentiable function  without  stationary points.  Given the positivity of partial derivatives, its maximum value will be attained at the limit as $\theta_i \rightarrow \infty$, $i=1,\ldots,r-1$:
\begin{equation}
\label{max_r}
\sup \left\{c(\theta), \theta\in\Theta_0^b\right\}= \lim_{\theta_1\rightarrow \infty, \ldots, \theta_{r-1}\rightarrow \infty}c(\theta) = pr\left(\sum_{i=r}^n X_i\geq r\right) +  pr\left(\sum_{i=r}^nY_i \geq 1\right).
\end{equation}
The value of the supremum does not depend on unknown parameters, since the last $n-r+1$ components of $\theta$ are  zero for $\theta \in \Theta_0^b$. 
Since $r> n-r+1$,  the first probability on the right hand-side in \eqref{max_r} is zero and therefore
$$
\sup \left\{c(\theta), \theta\in\Theta_0^b\right\} = 1-pr\left(\sum_{i=r}^n Y_i = 0\right) = 1-\Phi(t)^{n-r+1} = 1-\left(1-\frac{\alpha}{n-r+1}\right)^{n-r+1} \leq \alpha,
$$
 where the second and third equality follow from definitions of $Y_i$ and $t$, respectively.
The rightmost equality holds if and only if $n=r$.

It remains to be proved that for $\theta \in \Theta_0\setminus\Theta_0^b$, the value of Type I error cannot exceed that of the supremum on the boundary, i.e. $c(\theta) \leq \sup \left\{c(\theta), \theta\in\Theta_0^b\right\}$. We proceed by contradiction. Assume that there is a point $\theta^\ast \in \Theta_0\setminus \Theta_0^b$ such that $c(\theta^\ast) > C$, where $ C =  \sup \left\{c(\theta), \theta\in\Theta_0^b\right\}$.
Note that  $c(\theta)$ is still increasing in the first $r-1$ components, and for any  finite $\theta^\ast$, it can be increased  by letting $\theta_i \rightarrow \infty, i=1,\ldots,r-1$. Since $\theta^\ast \notin \Theta_0^b$, at least one of the   components of $\theta^\ast_{r:n}$  is negative.  The two probabilities on the right hand-side of \eqref{max_r} depend on $\theta_{r:n}$. In particular,  the probability of the first event $[\sum_{i=r}^nX_i \geq r]$ remains zero and the probability of the event $[\sum_{i=r}^n Y_i \geq 1]$ decreases, since, by assumption  $\theta_i <0$ for at least one $i$ in $\left\{r,\ldots,n\right\}$. As a consequence $ c(\theta^\ast) < C$. 
\end{proof}

\section{Directional replicabilty when $r$ is small}
In this section we present a counterexample that shows that for smaller values of $r$ a certain correction factor is necessary.

Consider $r$ such that $2 \leq r \leq (n+1)/2$. The main difference  with respect to the setting of the previous section  is  the possibility  of simultaneous presence of   $r-1$ positive and $r-1$ negative effects. Given that the problem is invariant with respect to the permutations of $\theta$, we can without loss of generality consider the null parameter space $\Theta_0$, where:
$$
 \Theta_0 = \left\{\theta\in \mathbb{R}^n:
\begin{array}{ll}
     \theta_i \geq 0, & \,i=1,\ldots,r-1; \\
    \theta_i \leq 0,  & \,i=r,\ldots,2r-2;\\
    \theta_i=0,  & \, i = 2r-1, \ldots,n.
\end{array}\right.
$$

Consider a test, that for a given $\alpha\in (0,1/2)$, rejects $H_{r/n}$ if  $\min\left\{p_{r/n}^-,          p_{r/n}^+\right\}\leq \alpha$ as before.
The two events leading to Type I error in \eqref{type1} are no longer disjoint and thus the probability of Type I error, after performing  simple set operations, can be expressed as
\begin{equation}
\label{smallr_type1}
c(\theta)= pr_\theta(T_{(n-r+1)}\geq t) + pr_{\theta}(T_{(r)}\leq - t,\, T_{(n-r+1)}\leq t).
\end{equation}
Analytical study of $c(\theta)$ is more challenging in the presence of quantities pertaining to   the joint distribution of $T_{(n-r+1)}$ and $T_{(r)}$, but we can evaluate $c$ in two special points of the null parameter space. Let $\theta^+$ represent the setting with $r-1$ strong positive effects:
$$\theta^+ = \left\{\begin{array}{ll}
              \theta_i^+ = \infty &\, i=1,\ldots,r-1,\\
              \theta_i^+=0 & \, i=r,\ldots, n.
         \end{array}\right.
$$
Let $\theta^{*}$ represent the setting with $r-1$ strong positive effects and $r-1$ strong negative effects:
$$ \tilde{\theta} = \left\{ 
\begin{array}{ll}
\tilde{\theta_i} =\infty, & \, i=1,\ldots,r-1, \\
\tilde{\theta_i} = -\infty, &  i=r,\ldots,2r-2,\\
\tilde{\theta_i} = 0, &  i= 2r-1, \ldots, n. 
\end{array}
\right.
$$
Then it can be shown that
\begin{eqnarray*}
&c(\theta^+) = 1-\left\{1-\Phi(t)\right\}^{n-r+1} + \sum_{k=r}^{n-r+1} {{n-r+1}\choose{k}}\left\{1-\Phi(t)\right\}^k\left\{2\Phi(t)-1\right\}^{n-r+1-k},\\
&c(\tilde{\theta}) = 1-\left\{2\Phi(t)-1\right\}^{n-2r+2}. 
\end{eqnarray*}
These probabilities can be evaluated numerically. Figure \ref{smallr} shows $c(\theta^+)$ and $c(\tilde{\theta})$ as a function of $r$ for $n=20$. The nominal level is $\alpha =0.1$. For $r \in \left\{2,\ldots,7\right\}$, the probability of Type I error at $\tilde{\theta}$ exceeds that of $\theta^+$, and more importantly exceeds $\alpha$. As a consequence, Type I error defined as  a supremum of $c(\theta)$ over $\Theta_0$ will exceed $\alpha$. Interestingly, the situation is reversed for $r \in \left\{8,9,10\right\}$,  where $c(\tilde{\theta})  < c(\theta^+)<\alpha$. Note that  the plot features values of $r$ up to 10, since for $r>10$, we have $r>(n+1)/2$, the  discordant setting is no longer possible  and we return to the setting of Theorem \ref{theorem_one}. 

\begin{figure}
\centering
\begin{tikzpicture}
\begin{axis}[width=10cm,height=7cm,
    xlabel={$r$},
    ylabel={$c(\theta)$},
    ytick = {0.05, 0.1, 0.15, 0.2},
    yticklabels = {0.05, 0.1, 0.15, 0.2},
    ytick scale label code/.code={},
    mark ="*",
    xmin =2, 
    xmax = 10,
    xtick pos = left,
    ytick pos = left,
    legend pos =south west,
  legend style = {draw = white}
    ]
    \addplot[black, mark = x, mark size = 3pt,  thick]
    coordinates { (2, 0.09948031) (3,0.09553569) (4,0.09543220) (5,0.09544652) (6,0.09546549) (7,0.09548723)
  (8,0.09551232) (9, 0.09554163) (10,0.09557629)};
    \addplot[black, mark = *, mark size = 3pt, thick]
coordinates{(2,0.17343603) (3,0.16370412) (4,0.15268471) (5,0.14010534) (6,0.12561125) (7,0.10873181) (8,0.08882939)
    (9,0.06501844) (10,0.03603306)};
    
    \addplot[mark=none, black, dashed] coordinates {(2,0.2) (3,0.2) (4,0.2) (5, 0.2) (6,0.2) (7,0.2) (8,0.2) (9, 0.2) (10,0.2)};
    \addplot[mark=none, black, dashed] coordinates {(2,0.1) (3,0.1) (4,0.1) (5, 0.1) (6,0.1) (7,0.1) (8,0.1) (9, 0.1) (10,0.1)};
  \addlegendentry{Concordant $\theta^+$};
  \addlegendentry{Discordant $\tilde{\theta}$};
   \end{axis}
\end{tikzpicture}
\caption{\label{smallr}Type I error probability as a function of $r$ for two specific parameter configurations: "concordant"  representing the presence of $r-1$ positive effects and "discordant" representing the presence of $r-1$  strong positive effects and $r-1$ strong negative effects. Dashed lines representing the nominal level of the test $\alpha =0.1$, as well as  the level of the corrected test $2\alpha$, are added for reference. The total number of studies is $n=20$.}
\end{figure}
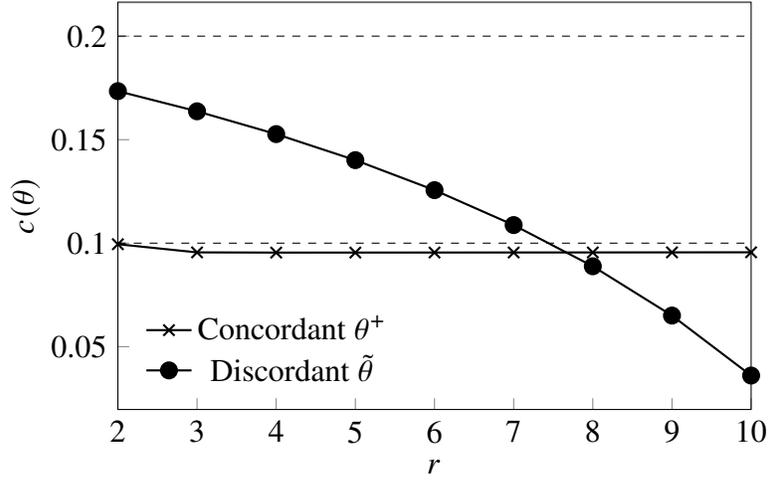

\begin{remark}  Theorem \ref{theorem_one} states then when $r>(n+1)/2$, the factor of 2 can be omitted, while the previous example illustrates that for some values of $r$   below this threshold, i.e. $r\in \left\{2,\ldots,7\right\}$, a correction factor is necessary.   For $r\in \left\{8,9,10\right\}$,  the Type I error at two considered points is below $\alpha$, so that it is possible that the Type I error is controlled for all  points of the null  parameter space. In other words, Theorem \ref{theorem_one} gives a sufficient but not necessary condition for  validity of the $p$-value  $p_{r/n}=\min\left\{p_{r/n}^-,          p_{r/n}^+\right\}$. For instance, it can be shown that $p_{r/n}$ remains a valid $p$-value for $n=3$ and $r=2$, see Proposition 1 in Supplementary material.  
\end{remark}

\section{Data adaptive choice of $r$}
We have so far assumed that $r$ is chosen prior to analysis. In many circumstances there are no substantive considerations indicating which value of $r$ should be preferred. In those situations, it may be desirable to determine $r$ in a 
data adaptive manner. This can be achieved by sequential testing of a collection of nested null hypotheses $H_{k/n}, H_{(k+1)/n}, \ldots H_{n/n}$, where   $k= \lceil (n+2)/2\rceil$ \citep{maurer1995multiple}. The procedure starts from  $H_{k/n}$ and tests each hypothesis at level $\alpha$. If the hypothesis is rejected, the testing proceeds to the next,  if not the procedure terminates.  Let $l$ be the index of the last rejected hypothesis, set to zero if $H_{k/n}$ is not rejected. Then $l$ is a $(1-\alpha)$ lower confidence bound for the maximum number of effects of the same sign, that is  $l \leq \max\left\{n^+, n^-\right\}$ with probability $1-\alpha$, or equivalently, there are at least $l$  effects in the same direction with probability $1-\alpha$.

\section{Discussion}

Bonferroni method is a simple way to combine $p$-values under arbitrary  dependence; here however we assume data coming from different studies to be independent. In that case, the power of the Bonferroni method can be improved by methods that exploit independence. In particular, it can be easily shown that for $r>(n+1)/2$, the result of Theorem  1 remains valid when the Bonferroni correction is substituted by the \v Sid\'ak correction that assumes independence \citep{vsidak1967rectangular}. 
Other combining functions that assume independence include the Simes method
$$
p_{r/n}^+= \min_{i=1,\ldots,n-r+1}\left\{\frac{n-r+1}{i}p_{(r-1+i)}\right\},
$$
and the Fisher combining function
$$
p_{r/n}^+ = pr\left\{\chi^2_{2(n-r+1)}\geq -2\sum_{i=r}^{n}\log p_{(i)}\right\}.
$$
An interesting question that awaits future research is whether the result presented in this work can be extended to these combining functions.

In this work, we have focused our attention on a single directional replicability hypothesis. In many applications, many features are studied simultaneously, and a collection of  replicability hypotheses is considered instead. In those situations, the power for identifying replicating signals can be increased by a careful consideration of the multiple testing aspect. Two existing approaches include filtering based on conditioning  \citep{wang2022detecting,dickhaus2021procedure}. To obtain directional replicability claims, one can follow \cite{owen2009karl} and apply proposed procedures  over the set of features twice, once for each direction at the significance level $\alpha/2$ and declare as replicated  features belonging to the union of the two rejection sets. An open question that awaits future research is whether the correction factor of 2 can be removed in this case.

The special case  of $n$ out $n$ replicability corresponds to testing whether all components of $\theta$ are of the same sign. This is also a special case of a problem  studied by \cite{sasabuchi1980test} in the context of testing hypotheses pertaining to multivariate normal means. The Author has shown that the test of Theorem \ref{theorem_one}, when $r=n$, is a likelihood ratio test for a slightly different null hypothesis, while \cite{berger1989uniformly} showed that the result of \cite{sasabuchi1980test}  remains valid for the null hypothesis $H_{n/n}$.   Our work can be seen as an extension of these results to $r<n$.

In general, rejection of 
$H_{r/n}$ allows one to conclude that   there are at least 
$r$ effects sharing a common sign, without indicating whether that sign is positive or negative. Inferring the sign from data after rejecting the null hypothesis can lead to Type III error,  an issue of concern in directional inference, see \cite{heller2024simultaneous}. It is easily checked that the procedure presented in Theorem \ref{theorem_one} controls Type III error  (Proposition 2 in Supplementary material)  and thus allows one to infer the sign of the replicated effects post-hoc: the sign is positive if $p_{r/n}^+ < p_{r/n}^-$, and  the sign  is negative otherwise.

\section{Acknowledgments}
This work was supported in part by United States National Institute on Aging grants R01AG080598 and  R03AG102038 and by the European Union, the National Recovery and Resilience Plan project “Age-It - Ageing well in an ageing society” (PE0000015 - CUP H73C22000900006). 

\bibliographystyle{apalike}
\bibliography{biblio}

\begin{thebibliography}{}

\bibitem[Benjamini and Heller, 2008]{benjamini2008screening}
Benjamini, Y. and Heller, R. (2008).
\newblock Screening for partial conjunction hypotheses.
\newblock {\em Biometrics}, 64(4):1215--1222.

\bibitem[Berger, 1989]{berger1989uniformly}
Berger, R.~L. (1989).
\newblock Uniformly more powerful tests for hypotheses concerning linear
  inequalities and normal means.
\newblock {\em Journal of the American Statistical Association},
  84(405):192--199.

\bibitem[Bogomolov and Heller, 2023]{bogomolov2023replicability}
Bogomolov, M. and Heller, R. (2023).
\newblock Replicability across multiple studies.
\newblock {\em Statistical Science}, 38(4):602--620.

\bibitem[Camerer et~al., 2016]{camerer2016evaluating}
Camerer, C.~F., Dreber, A., Forsell, E., Ho, T.-H., Huber, J., Johannesson, M.,
  Kirchler, M., Almenberg, J., Altmejd, A., Chan, T., et~al. (2016).
\newblock Evaluating replicability of laboratory experiments in economics.
\newblock {\em Science}, 351(6280):1433--1436.

\bibitem[Dickhaus et~al., 2021]{dickhaus2021procedure}
Dickhaus, T., Heller, R., Hoang, A.-T., and Rinott, Y. (2021).
\newblock A procedure for multiple testing of partial conjunction hypotheses
  based on a hazard rate inequality.
\newblock {\em arXiv preprint arXiv:2110.06692}.

\bibitem[Heller and Solari, 2024]{heller2024simultaneous}
Heller, R. and Solari, A. (2024).
\newblock Simultaneous directional inference.
\newblock {\em Journal of the Royal Statistical Society Series B: Statistical
  Methodology}, 86(3):650--670.

\bibitem[Ioannidis, 2005]{ioannidis2005most}
Ioannidis, J.~P. (2005).
\newblock Why most published research findings are false.
\newblock {\em PLoS medicine}, 2(8):e124.

\bibitem[Jaljuli et~al., 2023]{jaljuli2023quantifying}
Jaljuli, I., Benjamini, Y., Shenhav, L., Panagiotou, O.~A., and Heller, R.
  (2023).
\newblock Quantifying replicability and consistency in systematic reviews.
\newblock {\em Statistics in Biopharmaceutical Research}, 15(2):372--385.

\bibitem[Maurer, 1995]{maurer1995multiple}
Maurer, W. (1995).
\newblock Multiple comparisons in drug clinical trials and preclinical assays:
  a-priori ordered hypothesis.
\newblock {\em Biomed. Chem.-Pharm. Ind.}, 6:3--18.

\bibitem[{Open Science Collaboration}, 2015]{open2015estimating}
{Open Science Collaboration} (2015).
\newblock Estimating the reproducibility of psychological science.
\newblock {\em Science}, 349(6251):aac4716.

\bibitem[Owen, 2009]{owen2009karl}
Owen, A.~B. (2009).
\newblock Karl {P}earson's meta-analysis revisited.
\newblock {\em The Annals of Statistics}, 37(6B):3867--3892.

\bibitem[Samuels, 1965]{samuels1965number}
Samuels, S.~M. (1965).
\newblock On the number of successes in independent trials.
\newblock {\em The Annals of Mathematical Statistics}, 34:1272--1278.

\bibitem[Sasabuchi, 1980]{sasabuchi1980test}
Sasabuchi, S. (1980).
\newblock A test of a multivariate normal mean with composite hypotheses
  determined by linear inequalities.
\newblock {\em Biometrika}, 67(2):429--439.

\bibitem[{\v{S}}id{\'a}k, 1967]{vsidak1967rectangular}
{\v{S}}id{\'a}k, Z. (1967).
\newblock Rectangular confidence regions for the means of multivariate normal
  distributions.
\newblock {\em Journal of the American statistical association},
  62(318):626--633.

\bibitem[Wang et~al., 2022]{wang2022detecting}
Wang, J., Gui, L., Su, W.~J., Sabatti, C., and Owen, A.~B. (2022).
\newblock Detecting multiple replicating signals using adaptive filtering
  procedures.
\newblock {\em Annals of statistics}, 50(4):1890.

\bibitem[Wang, 1993]{wang1993number}
Wang, Y.~H. (1993).
\newblock On the number of successes in independent trials.
\newblock {\em Statistica Sinica}, 3:295--312.

\end{thebibliography}

\appendix \section{Supplementary material}

\begin{proposition}
Consider the setting of Theorem 1 and let $n=3$ and $r=2$.    Then $p_{r/n}$ is a valid $p$-value for $H_{r/n}$. 
\end{proposition}
\begin{proof}
When $n=3$ and $r=2$, consider $(\theta_1, \theta_2)^\top \in (0,\infty)\times (-\infty,0)$ and $\theta_3=0$. The probabilities  $ pr(\sum_{i=2}^nX_i=r-1)$ and $pr(\sum_{i=2}^nY_i=r-1)$, featured in the expression for partial derivative of $c(\theta)$, can be evaluated in a straightforward fashion:
 \begin{align*}
 pr(X_2+X_3 =1) &= \Phi(t)+\Phi(t+\theta_2) -2\Phi(t)\Phi(t+\theta_2),\\ pr(Y_2+Y_3=1) &= \Phi(t)+\Phi(t-\theta_2)-2\Phi(t)\Phi(t-\theta_2). 
\end{align*}
Define a function $h(u) = \Phi(t)+\Phi(t+u)-2\Phi(t)\Phi(t+u)$. Then following the proof of Theorem 1
\begin{align}
\label{first}
  \frac{{\rm d}c(\theta)}{{\rm d}\theta_1} &= -f(t+\theta_1)h(\theta_2)+f(t-\theta_1)h(-\theta_2).
  \end{align}
  Analogously, the  partial derivative with respect to $\theta_2$ is
  \begin{align}\label{second}
  \frac{{\rm d}c(\theta)}{{\rm d}\theta_2} &= -f(t+\theta_2)h(\theta_1)+f(t-\theta_2)h(-\theta_1).
  \end{align}
Stationary points can be identified by setting both partial derivatives to zero and finding a pair of values $(\theta_1^\ast, \theta_2^\ast)^\top \in (0,\infty)\times (-\infty,0)$ so  that it satisfies this system of equations.  From \eqref{first}:
$$
\theta_1^\ast = -\frac{1}{2t}\log\left[\frac{h(-\theta_2^\ast)}{h(\theta_2^\ast)}\right].
$$
From \eqref{second}
$$
\theta_2^\ast = -\frac{1}{2t}\log\left[\frac{h(-\theta_1^\ast)}{h(\theta_1^\ast)}\right].
$$
Let $g(u) = -(1/(2t))\log\left[h(-u)/h(u)\right]$. Then a stationary point should satisfy  $\theta_1^\ast = g\circ g(\theta_1^\ast)$, but it can be easily checked that no such $\theta_1 \in (0,\infty)$ exists, see Figure \ref{fig1}. 
As a consequence, the function $c(\theta)$ has no stationary points in $\Theta_0$ and to find its supremum we  investigate its limiting behavior. We have
\begin{align*}
\lim_{\theta_1 \rightarrow \infty, \theta_2 \rightarrow -\infty} c(\theta)&= pr(X_3 =1) +pr(Y_3=1)   = \alpha,\\
\lim_{\theta_1\rightarrow \infty, \theta_2=0}c(\theta)&= pr(X_2+X_3=2) +pr(Y_2+Y_3 \ge 1)\\
&=\Phi(-t)^2 + 1-\Phi(t)^2 \\
&= 2-2\Phi(t) \\
&=\alpha,
\end{align*}
where the last equality follows from the definition of $t$, see the proof of Theorem 1. 
Analogously, $\lim_{\theta_1=0, \theta_2\rightarrow-\infty}c(\theta)=\alpha$, and thus $\sup \left\{c(\theta), \theta\in \Theta_0\right\}=\alpha$, which completes the proof.
\end{proof}

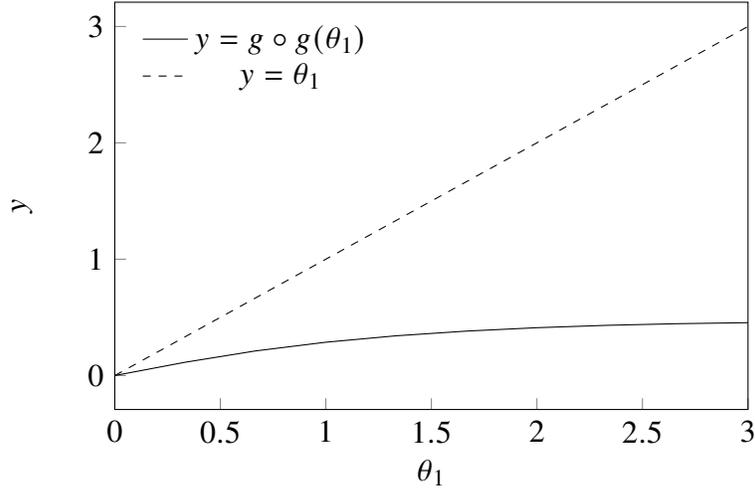
\begin{figure}
\centering
\begin{tikzpicture}
\begin{axis}[width=10cm,height=7cm,
    xlabel={$\theta_1$},
    ylabel={$y$},
    ytick scale label code/.code={},
    mark ="*",
    xmin =0, 
    xmax = 3,
    xtick pos = left,
    ytick pos = left,
    legend pos= north west,
    legend style = {draw = white}
    ]
    \addplot[black]
    coordinates { (0, 0) (0.3333333,0.1142378) (0.6666667, 0.2115821) (1.0000000, 0.2863999) (1.3333333, 0.3416454) (1.6666667, 0.3820159)
  (2.0000000, 0.4111723) (2.3333333, 0.4316986) (2.6666667,0.4455676) (3.0000000, 0.4544309)};
  \addplot[black,dashed] {x};
  \addlegendentry{$y=g\circ g(\theta_1)$};
  \addlegendentry{$y=\theta_1$}
   \end{axis}
   \end{tikzpicture}
   \caption{{\label{fig1}} Plot of the function $g\circ g(\theta_1)$. Identity function $y=\theta_1$ is added for reference.  }
\end{figure}

\begin{proposition} Consider the procedure presented in Theorem 1. Then Type III error is controlled at level $\alpha$ in the sense that 
$$
\sup_{\theta \in \Theta_0: n^+(\theta)\geq r}pr(T_{(r)}< -t)\leq \alpha \quad \mbox{and}\quad  \sup_{\theta \in \Theta_0: n^-(\theta)\geq r}pr(T_{(n-r+1)}> t)\leq \alpha,
$$
where $n^+(\theta)$ and $n^-(\theta)$ indicate the number of positive and negative components of $\theta$, respectively. 
\begin{proof}
    Due to the symmetry of the problem, it is sufficient to prove the first inequality. We have
    \begin{align*}
    \sup_{\theta \in \Theta_0: n^+(\theta)\geq r}pr(T_{(r)}< -t)& \leq \sup_{\theta \in \Theta_0: n^+(\theta)\geq r}\left\{pr(T_{(r)}< -t) + pr(T_{(n-r+1)}>t)\right\}\\
    & \leq \sup_{\theta \in \Theta_0}\left\{pr(T_{(r)}< -t) + pr(T_{(n-r+1)}>t)\right\}\\
    & =pr(\mbox{Type I error is made})\\
    &\leq \alpha,
    \end{align*}
where the last inequality follows from Theorem 1. 
\end{proof}
\end{proposition}

\end{document}